\documentclass[]{interact}

\usepackage{epstopdf}
\usepackage[caption=false]{subfig}

\usepackage[numbers,sort&compress]{natbib}
\usepackage{algorithm, algorithmicx, algpseudocode, tikz}
\bibpunct[, ]{[}{]}{,}{n}{,}{,}
\renewcommand\bibfont{\fontsize{10}{12}\selectfont}

\theoremstyle{plain}
\newtheorem{theorem}{Theorem}[section]
\newtheorem{lemma}[theorem]{Lemma}

\newtheorem{proposition}[theorem]{Proposition}

\theoremstyle{definition}
\newtheorem{definition}[theorem]{Definition}

\theoremstyle{remark}
\newtheorem{remark}{Remark}

\begin{document}

\title{Toward Solving 2-TBSG Efficiently}

\author{
\name{Zeyu Jia\textsuperscript{a}\thanks{CONTACT \newline Zeyu Jia. Email: jiazy@pku.edu.cn\newline Zaiwen Wen. Email: wenzw@math.pku.edu.cn\newline Yinyu Ye. Email: yinyu-ye@stanford.edu}, Zaiwen Wen\textsuperscript{b} and Yinyu Ye\textsuperscript{c}}
\affil{\textsuperscript{a}School of Mathematical Science, Peking University,
China; \textsuperscript{b}Beijing International Center for Mathematical
Research, Peking University, China; \textsuperscript{c}Department of Management
Science and Engineering, Stanford University, Stanford, CA, USA}
}

\maketitle

\begin{abstract}
2-TBSG is a two-player game model which aims to find Nash equilibriums and is widely utilized in reinforced learning and AI. Inspired by the fact that  the simplex method for solving the
deterministic discounted Markov decision processes (MDPs) is strongly polynomial
independent of the discounted factor, we are trying to answer an open
problem  whether there is a similar algorithm for 2-TBSG. We develop a simplex strategy iteration where one player updates its strategy with a simplex step while the other player finds an optimal
counterstrategy in turn, and a modified simplex strategy iteration. Both of them belong to a class of geometrically converging algorithms.
We establish the strongly polynomial property of these algorithms by considering
a strategy combined from the current strategy and the equilibrium strategy. Moreover, we present a method to transform general 2-TBSGs into special 2-TBSGs where each state has exactly two actions.
\end{abstract}


\section{Introduction}
 Markov decision process (MDP) is a widely used model in machine learning and operations research \cite{bellman1966dynamic}, which establishs basic rules of reinforcement learning. While solving an MDP focuses on maximizing (minimizing) the total reward (cost) for only one player, we consider a broader class of problems, the 2-player turn based stochastic games (2-TBSG) \cite{shapley1953stochastic}, which involves two players with opposite objectives. One player aims to maximize the total reward, and the other player aims to minimize the total reward. MDP and 2-TBSG have many useful applications, see \cite{howard1960dynamic, derman1970finite, puterman2014markov, bertsekas1995dynamic, filar2012competitive, neyman2003stochastic, silver2016mastering}.

 Similar to MDP, every 2-TBSG has its state set and action set, both of which are divided into two subsets for each player, respectively. Moreover, its transition probability matrix describes the transition distribution over the state set conditioned on the current action, and its reward function describes the immediate reward when taking the action.

We use a strategy (policy) to denote a mapping from the state set into the action set. In our setting, we focus on the discounted 2-TBSG, where the reward in later steps is multiplied by a discounted factor. Given strategies (policies) for both players, the total reward is defined to be the sum of all discounted rewards. We solve a 2-TBSG by finding its Nash equilibrium strategy (equilibrium strategy for short), where the first player cannot change its own strategy to obtain a larger total reward, and the second player cannot change its own strategy to obtain a smaller total reward. MDP can be viewed as a special case of 2-TBSG, where all states belong to the first player. In such cases, the equilibrium strategy agrees with the optimal policy of MDP.

MDPs have their linear programming (LP) formulations \cite{derman1970finite}. Hence algorithms solving LP problems can be used to solve MDPs. One of the most commonly used algorithm in MDP is the policy iteration algorithm \cite{howard1960dynamic}, which can be viewed as a parallel counterpart of the simplex method solving the corresponding LP. In paper \cite{10.2307/41412326}, both the simplex method solving the corresponding LP and the policy iteration algorithms have been proved to find the optimal policy in $\mathcal{O}\left(\frac{ml}{1-\gamma}\log\frac{l}{1-\gamma}\right)$, where $m, l, \gamma$ are the number of actions, the number of states and the discounted factor, respectively. Later in \cite{Hansen:2013:SIS:2432622.2432623}, the bound for the policy iteration algorithm is improved by a factor $l$ to $\mathcal{O}\left(\frac{m}{1-\gamma}\log\frac{l}{1-\gamma}\right)$. In \cite{scherrer2013improved}, this bound is improved to $\mathcal{O}\left(\frac{m}{1-\gamma}\log\frac{1}{1-\gamma}\right)$. When the MDP is deterministic (all transition probabilities are either $0$ or $1$), a strongly polynomial bound independent on the discounted factor is proved in \cite{Post:2013:SMS:2627817.2627922} for the simplex policy iteration method (each iteration changes only one action): $\mathcal{O}(m^{2}l^{3}\log^{2}l)$ for uniform discounted MDPs and $\mathcal{O}(m^{3}l^{5}\log^{2}l)$ for nonuniform discounted MDPs.

However, there is no simple LP formulation for 2-TBSGs. The strategy iteration algorithm \cite{rao1973algorithms}, an analogue to the policy iteration, is a commonly used algorithm in finding the equilibrium strategy of 2-TBSGs. It is a strongly polynomial time algorithm first proved in \cite{Hansen:2013:SIS:2432622.2432623} with a guarantee to find the equilibrium in $\mathcal{O}\left(\frac{m}{1-\gamma}\log\frac{l}{1-\gamma}\right)$ iterations if the discounted factor is fixed. When the discounted factor is not fixed, an exponential lower bound is given for the policy iteration in MDP \cite{fearnley2010exponential} and for the strategy iteration in 2-TBSG \cite{friedmann2009exponential}.
 It is an open problem whether there is a strongly  polynomial algorithm whose
 complexity is independent of the discounted factor for 2-TBSG.


Motivated by the strongly polynomial simplex algorithm for solving MDPs, we present a simplex strategy iteration algorithm and a modified
simplex strategy iteration algorithm for the 2-TBSG. In both algorithms each
player updates in turn, where the second player always finds the best
counterstrategy in its turn. In the simplex strategy iteration algorithm the
first player updates its strategy using the simplex algorithm.  In the modified
simplex strategy iteration algorithm, the first player updates the action
leading to the largest improvement after the second player finds the optimal
counterstrategy.  
When the second player is trivial, the 2-TBSG becomes an MDP and the simplex
strategy iteration algorithm can find its solution in strongly polynomial time
independent of the discounted factor, which is a property not possessed by the
strategy iteration algorithm in \cite{Hansen:2013:SIS:2432622.2432623}.

 We also develop a proof technique to prove the strongly polynomial complexity for a class of geometrically converging algorithms. This class of algorithms includes the strategy iteration algorithm, the simplex strategy iteration algorithm, and the modified simplex strategy iteration algorithm. The complexity for the strategy iteration algorithm given in \cite{Hansen:2013:SIS:2432622.2432623} can be recovered by our techniques. Our techniques use a combination of the current strategy and the equilibrium strategy. We establish a bound of ratio between the difference of value from the current strategy to the equilibrium strategy, and the difference of value from the combined strategy to the equilibrium strategy. Using this bound and the geometrically converging property, we can prove that after a certain number of iterations, one action will disappear forever, which leads to strongly polynomial convergence when the discount factor is fixed.
 Although
we have not fully answered the open progblem, 
 our algorithms and analysis point out a possible way for conquering the difficulities.

 Furthermore, 2-TBSG where each state has exactly two actions can be
transformed into a linear complementary problem
\cite{10.1007/978-3-540-69407-6_32}. An MDP where each state has exactly two
actions can be solved by a combinatorial interior point method \cite{ye2005new}.
In this paper we present a way to transform a general 2-TBSG into a 2-TBSG where
each state has exactly two actions. The number of states in this constructed
2-TBSG is $\tilde{\mathcal{O}}(m + l)$ (we use $\tilde{\mathcal{O}}$ to hide
log factors of $l, m$). This result enables the application of both results in
\cite{10.1007/978-3-540-69407-6_32, ye2005new} to general cases.

The rest of this paper is organized as follows. In Section \ref{sec:pre} we present some basic concepts and lemmas of the 2-TBSG. In Section \ref{sec:alg}  we describe the simplex strategy iteration algorithm and the modified simplex strategy iteration algorithm. The proof of complexity of the class of geometrically converging algorithm is given in Section \ref{sec:proof}. The transformation from general 2-TBSGs into special 2-TBSGs is introduced in Section \ref{sec:sp}.

\section{Preliminaries} \label{sec:pre}
\par In this section, we present some basic concepts of 2-TBSG. Our focus here is on the discounted 2-TBSG, defined as follows.
\begin{definition}
	A discounted 2-TBSG (2-TBSG for short) consists of a tuple $(\mathcal{S}, \mathcal{A}, P, r, \gamma)$, where $\mathcal{S} = \mathcal{S}_{1}\cup\mathcal{S}_{2}, \mathcal{A} = \mathcal{A}_{1}\cup\mathcal{A}_{2}$. $\mathcal{S}_{1}, \mathcal{S}_{2}, \mathcal{A}_{1}, \mathcal{A}_{2}$ are the state set and the action set of each player, respectively. $P\in\mathbb{R}^{|\mathcal{A}|\times|\mathcal{S}|}$ is the transition probability matrix, where $P(a, s)$ denotes the probability of the event that the next state is $s$ conditioned on the current action $a$. $r\in\mathbb{R}^{|\mathcal{A}|}$ is the reward vector, where $r_{a}$ denotes the immediate reward function received using action $a$. To be convenient, we use $m = |\mathcal{A}|$ to denote the number of actions, and $l = |\mathcal{S}|$ to denote the number of states.
\end{definition}

Given a state $s\in\mathcal{S}$ in 2-TBSG setting, we use $\mathcal{A}_{s}$ to denote the set of available actions corresponding to state $s$. A deterministic strategy (strategy for short) $\pi = (\pi_{1}, \pi_{2})$ is defined such that $\pi_{1}, \pi_{2}$ are mappings from $\mathcal{S}_{1}$ to $\mathcal{A}_{1}$ and from $\mathcal{S}_{2}$ to $\mathcal{A}_{2}$, respectively. Moreover, each state $s\in\mathcal{S}$ matches to an action in $\mathcal{A}_{s}$.

For a given strategy $\pi = (\pi_{1}, \pi_{2})$, we define the transition probability matrix $P_{\pi}\in\mathbb{R}^{l\times l}$ and reward function $r_{\pi}\in\mathbb{R}^{l}$ with respect to $\pi$. The $i$-th row of $P_{\pi}$ is chosen to be the row of action $\pi(i)$ in $P$, and the $i$-th element of $r_{\pi}$ is chosen to be the reward of action $\pi(i)$. It is easy to observe that the matrix $P_{\pi}$ is a stochastic matrix. We next define the value vector and the modified reward function.
\begin{definition}
	The value vector $v^{\pi}\in\mathbb{R}^{l}$ of a given strategy $\pi = (\pi_{1}, \pi_{2})$ is
	\begin{equation*}
		v^{\pi_{1}, \pi_{2}} = v^{\pi} = (I - \gamma P_{\pi})^{-1}r_{\pi}.
	\end{equation*}
\end{definition}
\begin{definition}
	The modified reward function $r^{\pi}\in\mathbb{R}^{m}$ of a given strategy $\pi$ is defined as
	\begin{equation*}
		r^{\pi} = r - (J - \gamma P)v_{\pi},
	\end{equation*}
	where $J\in\mathbb{R}^{m\times l}$ is defined as
	\begin{equation*}
		J_{ji} = \begin{cases} 1&\quad \text{if }j\in\mathcal{A}_{i},\\0&\quad \text{otherwise}.\end{cases}
	\end{equation*}
\end{definition}
\par Furthermore, for a given 2-TBSG, the optimal counterstrategy against
another player's given strategy is defined in Definition \ref{counter def}. The
equilibrium strategy is given in Definition \ref{equilibrium def}.
\begin{definition}\label{counter def}
	For player 2's strategy $\pi_{2}$, player 1's strategy $\pi_{1}$ is the optimal counterstrategy against $\pi_{2}$ if and only if for any strategy $\pi_{1}'$ of player 1, we have
	\begin{equation*}
		v^{\pi_{1}, \pi_{2}}\ge v^{\pi_{1}', \pi_{2}}.
	\end{equation*}
	Player 2's optimal counterstrategy can be defined similarly: $\pi_{2}$ is the optimal counterstrategy against $\pi_{1}$ if and only if for any strategy $\pi_{2}'$, $v^{\pi_{1}, \pi_{2}}\le v^{\pi_{1}, \pi_{2}'}$.
	Here for two value vector $v, v'$, we say $v\ge v'$ ($v\le v'$) if and only if $v(s)\ge v'(s)$ ($v(s)\le v'(s)$) for $\forall s\in\mathcal{S}$.
\end{definition}
\begin{definition}\label{equilibrium def}
	A strategy $\pi = (\pi_{1}, \pi_{2})$ is called an equilibrium strategy, if and only if $\pi_{1}$ is the optimal counterstrategy against $\pi_{2}$, and $\pi_{2}$ is the optimal counterstrategy against $\pi_{1}$.
\end{definition}
\par To describe the property of equilibrium strategies, we present Theorems \ref{existence} and \ref{optimal} given in \cite{Hansen:2013:SIS:2432622.2432623, shapley1953stochastic}. Theorem \ref{existence} indicates the existence of an equilibrium strategy.
\begin{theorem}\label{existence}
	Every 2-TBSG has at least an equilibrium strategy. If $\pi$ and $\pi'$ are two equilibrium strategies, then $v^{\pi} = v^{\pi'}$. Furthermore, for any player 1's strategy $\pi_{1}$ (or player 2's strategy $\pi_{2}$), there always exists a player 2's optimal counterstrategy $\pi_{2}$ against $\pi_{1}$ (player 1's optimal counterstrategy $\pi_{1}$ against $\pi_{2}$), and for any two optimal counterstrategy $\pi_{2}, \pi_{2}'$ ($\pi_{1}, \pi_{1}'$), we have $v^{\pi_{1}, \pi_{2}} = v^{\pi_{1}', \pi_{2}}$ ($v^{\pi_{1}, \pi_{2}} = v^{\pi_{1}, \pi_{2}'}$).
\end{theorem}
The next theorem points out a useful depiction of the value function at the equilibrium.
\begin{theorem}\label{optimal}
	Let $\pi^{*}$ be an equilibrium strategy for 2-TBSG. If $\pi_{1}$ is a strategy of player 1, and $\pi_{2}$ is player 2's optimal counterstrategy against $\pi_{1}$, then we have $v^{\pi^{*}}\ge v^{\pi_{1}, \pi_{2}}$. The equality holds if and only if $(\pi_{1}, \pi_{2})$ is an equilibrium strategy.
\end{theorem}

 We now define the flux vector of a given strategy $\pi$.
\begin{definition}
	The flux $x^{\pi}\in\mathbb{R}^{m}$ of a given strategy $\pi$ is defined as
	\begin{equation*}
		\begin{aligned}
			& (x^{\pi})_{\pi} = (I - \gamma P_{\pi})^{-T}\mathbf{1},\\
			& (x^{\pi})_{a} = 0, \quad \forall a\in\mathcal{A}, a\not\in \pi.
		\end{aligned}
	\end{equation*}
\end{definition}
Our next lemma presents bounds and conditions of the flux vector, and the
relationship among the value function, the flux vector and reduced costs. This
lemma and the following several lemmas can  be found in 
\cite{Hansen:2013:SIS:2432622.2432623}. To make the paper self-contained, we
briefly give their proofs.
\begin{lemma}\label{several}
	For any strategy $\pi$, we have
	\begin{enumerate}
		\item $\mathbf{1}^{T}x^{\pi} = \frac{l}{1-\gamma}$;
		\item for any $a\in \pi$, $1\le x_{a}\le \frac{l}{1-\gamma}$;
		\item $\mathbf{1}^{T}v^{\pi} = (x^{\pi})^{T}r$;
		\item $v^{\pi'} - v^{\pi} = (I - \gamma P_{\pi'})^{-1}(r^{\pi})_{\pi'}$, and moreover, $\mathbf{1}^{T}(v^{\pi'} - v^{\pi}) = (x^{\pi'})^{T}r^{\pi}$.
	\end{enumerate}
\end{lemma}
\begin{proof}
	Item (1) is proved by
	\begin{equation*}
		\mathbf{1}^{T}x^{\pi} = \mathbf{1}^{T}(I - \gamma P_{\pi})^{-T}\mathbf{1} = [(I - \gamma P_{\pi})^{-1}\mathbf{1}]^{T}\mathbf{1} = \frac{1}{1-\gamma}\mathbf{1}^{T}\mathbf{1} = \frac{l}{1-\gamma}. 
	\end{equation*}
	Item (2) is due to
	\begin{equation*}
		(x^{\pi})_{\pi} - \mathbf{1} = \gamma P_{\pi}(I - \gamma P_{\pi})^{-T}\mathbf{1}\ge 0.
	\end{equation*}
	This indicates that $(x^{\pi})_{a}\ge 1$, $\forall a\in\pi$. Hence we have $x^{\pi}\ge 0$ and $(x^{\pi})_{a}\le \frac{l}{1-\gamma}$ from item (1). Finally the last two items are obtained from
	\begin{equation*}
		\mathbf{1}^{T}v^{\pi} = \mathbf{1}^{T}(I - \gamma P_{\pi})^{-1}r_{\pi} = \mathbf{r}_{\pi}^{T}(I - \gamma P_{\pi})^{-T}\mathbf{1} = \mathbf{r}_{\pi}^{T}(x^{\pi})_{\pi} = (x^{\pi})^{T}r,
	\end{equation*}
	and
	\begin{equation*}
		\begin{aligned}
			v^{\pi'} - v^{\pi} & = (I - \gamma P_{\pi'})^{-1}r_{\pi'} - (I -
            \gamma P_{\pi'})^{-1}(I - \gamma P_{\pi'})v^{\pi};\\
			& = (I - \gamma P_{\pi'})^{-1}\left[r - (J - \gamma P)v_{\pi}\right]_{\pi'} = (I - \gamma P_{\pi'})^{-1}(r^{\pi})_{\pi'};\\
			\mathbf{1}^{T}(v^{\pi'} - v^{\pi}) & = \mathbf{1}^{T}(I - \gamma P_{\pi'})^{-1}(r^{\pi})_{\pi'} = (x^{\pi'})_{\pi'}^{T}(r^{\pi})_{\pi'} = (x^{\pi'})^{T}r^{\pi}.\\
		\end{aligned}
	\end{equation*}
\end{proof}
\par In the following, we present a lemma indicating the positiveness or
negativeness of the reduced costs of optimal counterstrategies and equilibrium strategies.
\begin{lemma}\label{condition}
	\begin{enumerate}
		\item A strategy $\pi_{1}$ for player 1 is an optimal counterstrategy against player 2's strategy $\pi_{2}$ if only if $(r^{\pi_{1}, \pi_{2}})_{\mathcal{A}_{1}}\le 0$. 
		\item A strategy $\pi_{2}$ for player 2 is an optimal counterstrategy against player 1's strategy $\pi_{1}$ if only if $(r^{\pi_{1}, \pi_{2}})_{\mathcal{A}_{2}}\ge 0$. 
		\item A strategy $\pi = (\pi_{1}, \pi_{2})$ is an equilibrium strategy if and only if it satisfies:
		\begin{equation*}
			(r^{\pi_{1}, \pi_{2}})_{\mathcal{A}_{1}}\le 0, \quad (r^{\pi_{1}, \pi_{2}})_{\mathcal{A}_{2}}\ge 0.
		\end{equation*}
	\end{enumerate}
\end{lemma}
\begin{proof}
	If $\pi_{1}, \pi_{2}$ satisfies $(r^{\pi_{1}, \pi_{2}})_{\mathcal{A}_{1}}\le 0$, then for any player 1's strategy $\pi_{1}'$, we have
	\begin{equation*}
		v^{\pi_{1}', \pi_{2}} - v^{\pi_{1}, \pi_{2}} = (I - \gamma P_{\pi_{1}', \pi_{2}})^{-1}(r^{\pi_{1}, \pi_{2}})_{\pi_{1}', \pi_{2}} = \sum_{n=0}^{\infty}\gamma^{n}P_{\pi_{1}', \pi_{2}}^{n}(r^{\pi_{1}, \pi_{2}})_{\pi_{1}', \pi_{2}}\le 0,
	\end{equation*}
	where the last inequality follows from $(r^{\pi_{1}, \pi_{2}})_{\pi_{1}'}\le 0$ for $\pi_{1}'\in\mathcal{A}_{1}$ and $(r^{\pi_{1}, \pi_{2}})_{\pi_{2}} = 0$.
	\par Suppose that player 1's strategy $\pi_{1}$ is the optimal counterstrategy against player 2's strategy $\pi_{2}$. For any $a'\in\mathcal{A}_{s}$, $s\in\mathcal{S}_{1}$ and $a'\not\in \pi_{1}$, we let 
	\begin{equation*}
		\pi_{1}'(s_{1}) = 
		\begin{cases}
			a' &\quad \text{if } s_{1} = s;\\
			\pi_{1}(s_{1}) &\quad \text{else}. 
		\end{cases}
	\end{equation*}
	Then again from  Lemma \ref{several} (4), we have
	\begin{equation*}
		x^{\pi'_{1}, \pi_{2}}_{a'}r^{\pi}_{a'} = \mathbf{1}^{T}(v^{\pi_{1}', \pi_{2}} - v^{\pi_{1}, \pi_{2}})\le 0,
	\end{equation*}
	where the inequality comes from the definition of equilibrium strategies. Since $a'\in\pi_{1}'$, we have $x^{\pi'_{1}, \pi_{2}}_{a'}\ge 1$, which indicates that $r^{\pi}_{a'}\le 0$. With this estimation and $r^{\pi}_{a} = 0$ for $\forall a\in\pi$, we have proved that $(r^{\pi})_{\mathcal{A}_{1}}\le 0$. Hence, item (1) is established, and the proof of item (2) is similar. Finally item (3) follows from items (1) and (2) directly.
\end{proof}

\section{Geometrically Converging Algorithms} \label{sec:alg}
\par Inspired by the simplex method solving the LP corresponding to the MDP and the strategy iteration algorithm given in \cite{Hansen:2013:SIS:2432622.2432623}, we propose a simplex strategy iteration (Algorithm \ref{simplex}) and a modified simplex strategy iteration algorithm (Algorithm \ref{modified}) for 2-TBSG.

The simplex strategy iteration algorithm can be viewed as a generalization of
the strongly polynomial simplex algorithm in solving MDPs
\cite{Post:2013:SMS:2627817.2627922}. In our algorithm, both players update
their strategies in turn. In each iteration, while the first player updates its
strategy using the simplex method, which means only updating the action with the
largest reduced cost, the second player updates its strategy according to the optimal
counterstrategy. 
 When the second player has only one possible action and the transition matrix
 is deterministic, the 2-TBSG reduces to a deterministic MDP. Then the simplex strategy iteration
 algorithm can find an equilibrium (optimal) strategy in strongly polynomial
 time independent of $\gamma$, which is a property has not been proven for the strategy iteration \cite{Hansen:2013:SIS:2432622.2432623}.

As for the modified simplex strategy iteration algorithm, it can be viewed as a modification of the simplex strategy iteration algorithm. In this algorithm, both players also update their strategies in turn, and the second player always finds the optimal counterstrategies in its moves. However, in each of the first player's move, only the action is updated which leads to the biggest improvement on the value function when the second player uses the optimal counterstrategy.

It is easy to know that every iteration of the simplex strategy iteration algorithm involves a step of a simplex update and a solution to an MDP. And every iteration of the modified simplex strategy iteration algorithm involves solutions to multiple MDPs. Hence every iteration in both of these two algorithms can be solved in strongly polynomial time when the discounted factor is fixed.

\begin{algorithm}
	\caption{A Simplex Strategy Iteration Method}
	\label{simplex}
	\begin{algorithmic}[1]
		\State \textbf{Initialize: } $\pi_{1}^{0}$ for player 1, $\pi_{2}^{0}$ for player 2, $n\leftarrow 0$.
		\Repeat
			\State Find $s\in\mathcal{S}_{1}$ and $a\in\mathcal{A}_{s}$ such that $r_{a}^{\pi_{1}^{n}, \pi_{2}^{n}}$ is the largest among such $(s, a)$.
			\State $\pi_{1}^{n+1}(s)\leftarrow a$, $\pi_{1}^{n+1}(s')\leftarrow \pi_{1}^{n}(s')$ for $\forall s'\in\mathcal{S}_{1}, s'\neq s$.
			\State Find the optimal counterstrategy $\pi_{2}^{n+1}$ against $\pi_{1}^{n+1}$.
		\Until{$\pi_{1}^{n} = \pi_{1}^{n+1}$.}
		\State \textbf{Output: } $\pi_{1}^{n}, \pi_{2}^{n}$.
	\end{algorithmic}
\end{algorithm}

\begin{algorithm}
	\caption{A Modified Simplex Strategy Iteration Algorithm}
	\label{modified}
	\begin{algorithmic}[1]
		\State \textbf{Initialize: } $\pi_{1}$ for player 1, $\pi_{2}$ for player 2
		\State Let $\pi_{1}^{0}\leftarrow \pi_{1}, n\leftarrow 0$
		\Repeat
			\State $\pi_{1}^{n+1}\leftarrow \pi_{1}^{n}, \quad \pi_{2}^{n+1}\leftarrow \pi^{n}_{2}$
			\For{any $s\in\mathcal{S}$ and action $a\in\mathcal{A}_{s}$}
				\State Let $\tilde{\pi}_{1}^{n+1}$ be the player 1's strategy where only state $s$'s action changes to $a$, and other states' actions keep to be the same as $\pi_{1}^{n}$.
				\State Let $\tilde{\pi}_{2}^{n+1}$ be the optimal counterstrategy against $\tilde{\pi}_{1}^{n+1}$.
				\If{$\mathbf{1}^{T}v^{\pi_{1}^{n+1}, \pi_{2}^{n+1}}\le \mathbf{1}^{T}v^{\tilde{\pi}_{1}^{n+1}, \tilde{\pi}_{2}^{n+1}}$}
					\State $\pi_{1}^{n+1}\leftarrow \tilde{\pi}_{1}^{n+1}, \quad \pi_{2}^{n+1}\leftarrow \tilde{\pi}^{n+1}_{2}$.
				\EndIf
			\EndFor
		\Until{$\mathbf{1}^{T}\pi_{1}^{n} = \mathbf{1}^{T}\pi_{1}^{n+1}$.}
		\State \textbf{Output: } $\pi_{1}^{n}, \pi_{2}^{n}$.
	\end{algorithmic}
\end{algorithm}

\par Next we present a class of geometrically converging algorithms used for proving the strongly polynomial complexity for several algorithms in the next section.
\begin{definition}
	We say a strategy-update algorithm (algorithms which update strategies for both players in each iteration) is a geometrically converging algorithm with parameter a $M$, if it updates a strategy $\pi^{n} = (\pi_{1}^{n}, \pi_{2}^{n})$ to $\pi^{n+1} = (\pi_{1}^{n+1}, \pi_{2}^{n+1})$ such that the following properties holds.
	\begin{itemize}
		\item $\pi_{2}^{n+1}$ is the optimal counterstrategy against $\pi^{n+1}_{1}$; 
		\item $(r^{\pi^{n}})_{\pi_{1}^{n+1}}\ge 0$;
		\item If $\mathbf{1}^{T}(v^{\pi^{n+1}} - v^{\pi^{n}}) = 0$, then $\pi^{n}$ is an equilibrium strategy;
		\item The updates of this algorithm satisfies
		\begin{equation*}
			\mathbf{1}^{T}\left(v^{\pi_1^*, \pi_2^*} - v^{\pi_1^{n+M}, \pi_2^{n+M}}\right)\le\frac{(1-\gamma)^{2}}{n^{2}}\cdot\mathbf{1}^{T}\left(v^{\pi_1^{*}, \pi_2^{*}} - v^{\pi_1^n, \pi_2^n}\right).
		\end{equation*}
	\end{itemize}
\end{definition} 

\par To begin with, we exhibit a lemma indicating the geometrically converging property of the value function in the simplex strategy iteration algorithm.
\begin{lemma}\label{simplex lemma}
	Suppose the sequence of strategy generated by the simplex strategy iteration algorithm is $\pi^{1} = (\pi^{1}_{1}, \pi^{1}_{2}), \pi^{2} = (\pi^{2}_{1}, \pi^{2}_{2}) \cdots, \pi^{n} = (\pi^{n}_{1}, \pi^{n}_{2}), \cdots$. Then the following inequality holds
	\begin{equation}\label{eq1}
		\mathbf{1}^{T}(v^{\pi^{*}} - v^{\pi^{n+1}})\le \left(1 - \frac{1-\gamma}{l}\right)\mathbf{1}^{T}(v^{\pi^{*}} - v^{\pi^{n}}).
	\end{equation}
\end{lemma}

\begin{proof}
	According to Algorithm \ref{simplex}, we have
\begin{equation*}
	\begin{aligned}
		\mathbf{1}^{T}(v^{\pi^{n+1}} - v^{\pi^{n}}) & \ge r^{\pi^{n}}_{a_{1}}x^{\pi^{n+1}}_{a_{1}}\ge r^{\pi^{n}}_{a_{1}}\ge \frac{1-\gamma}{l}\sum_{a\in\mathcal{A}_{1}}r^{\pi^{n}}_{a}x^{\pi^{n+1}}_{a}\\
		& \ge \frac{1-\gamma}{l}\sum_{a\in\mathcal{A}}r^{\pi^{n}}_{a}x^{\pi^{n+1}}_{a} = \frac{1-\gamma}{l}\mathbf{1}^{T}(v^{\pi^{*}} - v^{\pi^{n}}),
	\end{aligned}
\end{equation*}
where the second and third inequalities follow from Lemma \ref{several} (2) and the choice of $a_{1} = \arg\max_{a\in\mathcal{A}_{1}}r^{\pi^{n}}_{a}$, the fourth inequality follows from Lemma \ref{condition}, and the first inequality and last equation are due to Lemma \ref{several} (4) and Lemma \ref{condition}.
\end{proof}

\par Using this lemma, we show in the next proposition that the strategy iteration algorithm, Algorithm \ref{simplex} and Algorithm \ref{modified} all belong to the class of geometrically converging algorithms.
\begin{proposition}\label{prop}
	\begin{enumerate}
		\item The strategy iteration algorithm given in \cite{Hansen:2013:SIS:2432622.2432623} is a geometrically converging algorithm with parameter $M = \mathcal{O}\left(\frac{1}{1-\gamma}\log\frac{l}{1-\gamma}\right)$;
		\item The simplex strategy iteration algorithm (Algorithm \ref{simplex}) is a geometrically converging algorithm with parameter $M = \mathcal{O}\left(\frac{l}{1-\gamma}\log\frac{l}{1-\gamma}\right)$;
		\item The modified simplex strategy iteration algorithm (Algorithm \ref{modified}) is a geometrically converging algorithm with parameter $M = \mathcal{O}\left(\frac{l}{1-\gamma}\log\frac{l}{1-\gamma}\right)$;
	\end{enumerate}
\end{proposition}

\begin{proof}
	\par It is easy to verify that the previous described three algorithms satisfy the first three conditions in the definition of geometrically converging algorithms. Next, we prove that all of these algorithms satisfy the last condition. For the strategy iteration algorithm, according to Lemma 4.8 and Lemma 5.4 given in \cite{Hansen:2013:SIS:2432622.2432623}, we have
	\begin{equation*}
		\mathbf{1}^{T}(v^{\pi^{*}} - v^{\pi^{n+1}})\le \gamma\mathbf{1}^{T}(v^{\pi^{*}} - v^{\pi^{n}}).
	\end{equation*}
	Hence if $M = \frac{2c_{1}}{1-\gamma}\log\frac{l}{1-\gamma} = \mathcal{O}\left(\frac{1}{1-\gamma}\log\frac{l}{1-\gamma}\right)$ ($c_{1}\ge 1$ is a constant), then we obtain
	\begin{equation*}
		\begin{aligned}
			\mathbf{1}^{T}(v^{\pi^{*}} - v^{\pi^{n+M}}) & \le \gamma^{M}\mathbf{1}^{T}(v^{\pi^{*}} - v^{\pi^{n}})\le \gamma^{-2\log_{\gamma}\frac{n}{1-\gamma}}\mathbf{1}^{T}(v^{\pi^{*}} - v^{\pi^{n}})\\
			& = \frac{(1-\gamma)^{2}}{l^{2}}\mathbf{1}^{T}(v^{\pi^{*}} - v^{\pi^{n}}),
		\end{aligned}
	\end{equation*}
	and the last condition of geometrically converging algorithms is verified.
	\par For the simplex strategy iteration algorithm,  if we choose $M = \mathcal{O}\left(\frac{l}{1-\gamma}\log\frac{l}{1-\gamma}\right)$ ($c_{2}\ge 1$ is a constant), then according to inequality \eqref{eq1} we have
	\begin{equation*}
		\mathbf{1}^{T}(v^{\pi^{*}} - v^{\pi^{n+M}})\le \frac{(1-\gamma)^{2}}{l^{2}}\mathbf{1}^{T}(v^{\pi^{*}} - v^{\pi^{n}}),
	\end{equation*}
	and the last condition of geometrically converging algorithms is verified.
	\par Finally we consider the modified simplex strategy iteration algorithm. For $n\ge 2$, let $a_{1} = \arg\max_{a\in\mathcal{A}_{1}}r^{\pi^{n}}_{a}$, where $a_{1}$ is an action of state $s_{1}$. Let 
	\begin{equation*}
		\pi_{1}'(s) = \begin{cases}
			a_{1}, &\quad \text{if }s = s_{1},\\
			\pi^{n}(s), &\quad \text{others},
		\end{cases}
	\end{equation*}
	$\pi_{2}'$ be player 2's optimal counterstrategy against $\pi_{1}'$, and $\pi' = (\pi_{1}', \pi_{2}')$. Then from inequality \eqref{eq1}, we have
	\begin{equation*}
		\mathbf{1}^{T}(v^{\pi^{*}} - v^{\pi'})\le \left(1 - \frac{1-\gamma}{n}\right)\mathbf{1}^{T}(v^{*} - v^{\pi^{n}}).
	\end{equation*}
	According to Algorithm \ref{modified}, we have
	\begin{equation*}
		\mathbf{1}^{T}v^{\pi^{n+1}}\ge\mathbf{1}^{T}v^{\pi'},
	\end{equation*}
	which leads to the following estimation:
	\begin{equation*}
		\mathbf{1}^{T}(v^{\pi^{*}} - v^{\pi^{n+1}})\le \left(1 - \frac{1-\gamma}{n}\right)\mathbf{1}^{T}(v^{*} - v^{\pi^{n}}).
	\end{equation*}
	Therefore, similar to the previous case we can choose $M = \mathcal{O}\left(\frac{l}{1-\gamma}\log\frac{l}{1-\gamma}\right)$ such that
	\begin{equation*}
		\mathbf{1}^{T}(v^{\pi^{*}} - v^{\pi^{n+M}})\le \frac{(1-\gamma)^{2}}{l^{2}}\mathbf{1}^{T}(v^{\pi^{*}} - v^{\pi^{n}}),
	\end{equation*}
	and the last condition of geometrically converging algorithms is verified.
\end{proof}

\section{Strongly Polynomial Complexity of Geometrically Converging Algorithms} \label{sec:proof}
 In this section, we develop the strongly polynomial property of geometric
 converging algorithms if the parameter $M$ is viewed as a constant.
 Slightly different from the proof in \cite{Hansen:2013:SIS:2432622.2432623}
 for the strategy $(\pi_{1}^{n}, \pi_{2}^{n})$ at the $n$-th iteration, we
 present a proof by considering the strategy $(\pi_{1}^{n}, \pi_{2}^{*})$, where
 $(\pi_{1}^{*}, \pi_{2}^{*})$ is an equilibrium strategy. We show that
 $\mathbf{1}^{T}(v^{\pi_{1}^{*}, \pi_{2}^{*}} - v^{\pi_{1}^{n}, \pi_{2}^{*}})$
 can be both upper and lower bounded by some proportion of
 $\mathbf{1}^{T}(v^{\pi_{1}^{*}, \pi_{2}^{*}} - v^{\pi_{1}^{n}, \pi_{2}^{n}})$.
 By applying the property of geometrically converging algorithms, we obtain that after a certain number of iterations, a player 1's action will disappear in $\pi_{1}^{n}$ forever.
\begin{theorem}\label{convergence}
	Any geometrically converging algorithm with a parameter $M$ finds the equilibrium strategy in
	\begin{equation*}
		\mathcal{O}(Mm)
	\end{equation*}
	number of iterations.
\end{theorem}
\begin{proof}Suppose $\pi^{1} = (\pi_{1}^{1}, \pi_{2}^{1}), \pi^{2} = (\pi_{1}^{2}, \pi_{2}^{2}), \cdots, \pi^{n} = (\pi_{1}^{n}, \pi_{2}^{n})$ is the sequence generated by a geometrically converging algorithm. We define $\eta^{n} = (\pi_{1}^{n}, \pi_{2}^{*})$, where $\pi^{*} = (\pi_{1}^{*}, \pi_{2}^{*})$ is one of the equilibrium strategy.
\par According to Lemma \ref{condition} and the fact that $\pi_{2}^{n+1}$ is the optimal counterstrategy against $\pi_{1}^{n+1}$, and the definition of geometrically converging algorithm, we have
\begin{equation*}
	\mathbf{1}^{T}(v^{\pi^{n+1}} - v^{\pi^{n}}) = \sum_{a\in\pi_{1}^{n+1}}x^{\pi^{n+1}}_{a}r^{\pi^{n}}_{a} + \sum_{a\in\pi_{2}^{n+1}}x^{\pi^{n+1}}_{a}r^{\pi^{n}}_{a}\ge \sum_{a\in\pi_{1}^{n+1}}x^{\pi^{n+1}}_{a}r^{\pi^{n}}_{a}\ge 0,
\end{equation*}
which directly leads to
\begin{equation}\label{ineq11}
	\mathbf{1}^{T}v^{\pi^{n}}\le \mathbf{1}^{T}v^{\pi^{n+1}}.
\end{equation}
According to Lemma \ref{condition}, we have
\begin{equation*}
	\begin{aligned}
		& \mathbf{1}^{T}(v^{\pi^{*}} - v^{\eta^{n}}) = \mathbf{1}^{T}(v^{\pi_{1}^{*}, \pi_{2}^{*}} - v^{\pi_{1}^{n}, \pi_{2}^{*}}) = - (x^{\pi_{1}^{n}, \pi_{2}^{*}})^{T}r^{\pi_{1}^{*}, \pi_{2}^{*}} = - (x^{\pi_{1}^{n}, \pi_{2}^{*}})_{\mathcal{A}_{1}}^{T}r^{\pi_{1}^{*}, \pi_{2}^{*}}_{\mathcal{A}_{1}}\ge 0,\\
		& \mathbf{1}^{T}(v^{\eta^{n}} - v^{\pi^{n}}) = \mathbf{1}^{T}(v^{\pi_{1}^{n}, \pi_{2}^{*}} - v^{\pi_{1}^{n}, \pi_{2}^{n}}) = (x^{\pi_{1}^{n}, \pi_{2}^{*}})^{T}r^{\pi_{1}^{n}, \pi_{2}^{n}} = (x^{\pi_{1}^{n}, \pi_{2}^{*}})_{\mathcal{A}_{1}}^{T}r^{\pi_{1}^{n}, \pi_{2}^{n}}_{\mathcal{A}_{1}}\ge 0,
	\end{aligned}
\end{equation*}
which implies
\begin{equation}\label{ineq12}
	\mathbf{1}^{T}v^{\pi^{n}}\le \mathbf{1}^{T}v^{\eta^{n}}\le \mathbf{1}^{T}v^{\pi^{*}}.
\end{equation}
\par We next prove the following inequality:
\begin{equation}\label{ineq2}
	\mathbf{1}^{T}(v^{\pi^{*}} - v^{\eta^{n}})\ge\frac{1-\gamma}{n}\cdot\mathbf{1}^{T}(v^{\pi^{*}} - v^{\pi^{n}}).
\end{equation}
A direct calculation gives
\begin{equation*}
	\begin{aligned}
		\mathbf{1}^{T}(v^{\pi^{*}} - v^{\pi^{n}}) & = \mathbf{1}^{T}(v^{\pi_{1}^{*}, \pi_{2}^{*}} - v^{\pi_{1}^{n}, \pi_{2}^{n}}) = - (x^{\pi_{1}^{n}, \pi_{2}^{n}})^{T}r^{\pi_{1}^{*}, \pi_{2}^{*}}\\
		&  = - \sum_{a\in\pi^{n}_{1}}x^{\pi_{1}^{n}, \pi_{2}^{n}}_{a}r^{\pi_{1}^{*}, \pi_{2}^{*}}_{a} - \sum_{a\in\pi^{n}_{2}}x^{\pi_{1}^{n}, \pi_{2}^{n}}_{a}r^{\pi_{1}^{*}, \pi_{2}^{*}}_{a} \le - \sum_{a\in\pi^{n}_{1}}x^{\pi_{1}^{n}, \pi_{2}^{n}}_{a}r^{\pi_{1}^{*}, \pi_{2}^{*}}_{a},
	\end{aligned}
\end{equation*}
where the last inequality is obtained from Lemma \ref{condition}. Then noticing that
\begin{equation*}
	1\le x^{\pi_{1}^{n}, \pi_{2}^{n}}_{a}, x^{\pi_{1}^{n}, \pi_{2}^{*}}_{a}\le \frac{1-\gamma}{l}, \quad r^{\pi_{1}^{*}, \pi_{2}^{*}}_{a}\le 0, \quad \forall a\in \pi_{1}^{n},
\end{equation*}
we have
\begin{equation*}
	\begin{aligned}
		\mathbf{1}^{T}(v^{\pi^{*}} - v^{\eta^{n}}) & = \mathbf{1}^{T}(v^{\pi_{1}^{*}, \pi_{2}^{*}} - v^{\pi_{1}^{n}, \pi_{2}^{*}}) = - (x^{\pi_{1}^{n}, \pi_{2}^{*}})^{T}r^{\pi_{1}^{*}, \pi_{2}^{*}} = - \sum_{a\in\pi^{n}_{1}}x^{\pi_{1}^{n}, \pi_{2}^{*}}_{a}r^{\pi_{1}^{*}, \pi_{2}^{*}}_{a}\\
		& \ge - \frac{1-\gamma}{l}\sum_{a\in\pi^{n}_{1}}x^{\pi_{1}^{n}, \pi_{2}^{n}}_{a}r^{\pi_{1}^{*}, \pi_{2}^{*}}_{a}\ge \frac{1-\gamma}{l}\mathbf{1}^{T}(v^{\pi^{*}} - v^{\pi^{n}}).
	\end{aligned}
\end{equation*}
Then the inequality \eqref{ineq2} is proved.
\par Finally, we prove that for any $n$, either there exists an action $a_{1}$ in $\pi_{1}^{n}$ will never belong to $\pi_{1}^{n+m}$ when $m > M$, or we have 
\begin{equation*}
	\mathbf{1}^{T}(v^{\pi^{n+M+1}} - v^{\pi^{n+M}}) = 0.
\end{equation*}
Actually for any $p > M$, suppose $\mathbf{1}^{T}(v^{\pi^{n+M+1}} - v^{\pi^{n+M}}) \neq 0$, we obtain
\begin{equation*}
	\mathbf{1}^{T}(v^{\pi^{*}} - v^{\pi^{n+p}}) < \mathbf{1}^{T}(v^{\pi^{*}} - v^{\pi^{n+M}})\le \frac{(1-\gamma)^{2}}{l^{2}}\mathbf{1}^{T}(v^{\pi^{*}} - v^{\pi^{n}})
\end{equation*}
from \eqref{ineq11} and the definition of geometrically converging algorithm. Hence according to \eqref{ineq12} and \eqref{ineq2}, we get
\begin{equation}\label{ineq3}
	\mathbf{1}^{T}(v^{\pi^{*}} - v^{\eta^{n+p}})\le \mathbf{1}^{T}(v^{\pi^{*}} - v^{\pi^{n+p}}) < \frac{(1-\gamma)^{2}}{l^{2}}\mathbf{1}^{T}(v^{\pi^{*}} - v^{\pi^{n}})\le \frac{1-\gamma}{l}\mathbf{1}^{T}(v^{\pi^{*}} - v^{\eta^{n}}).
\end{equation}
Therefore, choosing $a_{1} = \arg\min_{a\in\pi_{1}^{n}}r_{a}^{\pi^{*}} \le 0$, and because for any $a\in\pi_{1}^{n}$, $r_{a}^{\pi^{*}} \le 0$ according to Lemma \ref{condition}, we obtain
\begin{equation*}
	\mathbf{1}^{T}(v^{\pi^{*}} - v^{\eta^{n}}) = -\sum_{a\in\pi_{1}^{n}}x_{a}^{\eta^{n}}r_{a}^{\pi^{*}}\le\left(\sum_{a\in\pi_{1}^{n}}x_{a}^{\eta^{n}}\right)\cdot (-r_{a_{1}}^{\pi^{*}}) \le -\frac{l}{1-\gamma}\cdot r_{a_{1}}^{\pi^{*}}
\end{equation*}
from Lemma \ref{several}. If $a\in\pi_{1}^{n+p}$, we have
\begin{equation*}
	\mathbf{1}^{T}(v^{\pi^{*}} - v^{\eta^{n+p}}) = -\sum_{a\in\pi_{1}^{n+p}}x_{a}^{\eta^{n+p}}r_{a}^{\pi^{*}}\ge - x_{a_{1}}^{\eta^{n+p}}r_{a_{1}}^{\pi^{*}}\ge - r_{a_{1}}^{\pi^{*}},
\end{equation*}
where the first inequality is due to Lemma \ref{condition} and the second inequality is due to Lemma \ref{several}. Therefore, combining these two inequalities and the inequality \eqref{ineq3} and noticing that $r_{a_{1}}^{\pi^{*}}\le 0$, we get
\begin{equation*}
	- r_{a_{1}}^{\pi^{*}}\le \mathbf{1}^{T}(v^{\pi^{*}} - v^{\eta^{n+p}}) < \frac{1-\gamma}{l}\mathbf{1}^{T}(v^{\pi^{*}} - v^{\eta^{n}})\le - \frac{1-\gamma}{l}\cdot\frac{l}{1-\gamma}\cdot r_{a_{1}}^{\pi^{*}} = -  r_{a_{1}}^{\pi^{*}}.
\end{equation*}
This leads to contradiction.
\par The previous derivation means that if $\mathbf{1}^{T}(v^{\pi^{n+M+1}} -
v^{\pi^{n+M}}) = 0$ does not hold for $n$, then an action of $\pi^{n}$ must
disappear after $\pi^{n+M}$ forever. Hence every after $M$ iterations an action
will disappear forever. This process cannot happen for more than $m - l$ times
(since there are $m$ actions and every strategy has $n$ actions), which indicates that for some $n > M(m-l)$,
\begin{equation*}
	\mathbf{1}^{T}(v^{\pi^{n+M+1}} - v^{\pi^{n+M}}) = 0.
\end{equation*}
It follows from the definition of geometrically converging algorithm that $\pi^{n+M}$ is the equilibrium strategy. This indicates that within
\begin{equation*}
	\mathcal{O}(mM)
\end{equation*}
number of iterations, we can find one of the equilibrium strategies.
\end{proof}

\par Our next theorem presents the complexity of the strategy iteration algorithm, the simplex strategy iteration algorithm and the modified simplex strategy iteration algorithm.
\begin{theorem}
	\par The following algorithms has strongly polynomial convergence when the discounted factor is fixed.
	\begin{itemize}
		\item The strategy iteration algorithm given in \cite{Hansen:2013:SIS:2432622.2432623} can find the equilibrium strategy within $\mathcal{O}\left(\frac{m}{1-\gamma}\log\frac{l}{1-\gamma}\right)$ iterations;
		\item The simplex strategy iteration algorithm (Algorithm \ref{simplex}) can find the equilibrium strategy within $\mathcal{O}\left(\frac{ml}{1-\gamma}\log\frac{l}{1-\gamma}\right)$ iterations;
		\item The modified simplex strategy iteration algorithm (Algorithm \ref{modified}) can find the equilibrium strategy within $\mathcal{O}\left(\frac{ml}{1-\gamma}\log\frac{l}{1-\gamma}\right)$ iterations;
	\end{itemize}
\end{theorem}
\begin{proof}
	The proof of this theorem directly follows from Theorem \ref{convergence} and Proposition \ref{prop}.
\end{proof}
\begin{remark}
	It is easy to note that the terminated condition of the simplex strategy iteration algorithm and the modified simplex strategy iteration algorithm is equivalent to the condition of meeting an equilibrium strategy. Hence the above theorem also indicates that these two algorithms terminate within $\mathcal{O}\left(\frac{ml}{1-\gamma}\log\frac{l}{1-\gamma}\right)$ iterations.
\end{remark}

\section{Transform General 2-TBSGs into Special 2-TBSGs}  \label{sec:sp}
\par We prove in this section that every 2-TBSG can be transformed into a new 2-TBSG where each state has exactly two actions. A formal description is given in the next theorem.
\begin{theorem}\label{2-TBSG}
	Given a 2-TBSG with $l$ states and $m$ actions whose  state set is $\mathcal{S}$, we can construct a new 2-TBSG with state set $\mathcal{S}'$ satisfying the following properties.
	\begin{itemize}
		\item The number of states in the constructed 2-TBSG is bounded by a polynomial of $m$ and $l$:
		\begin{equation}\label{eqin1}
			|\mathcal{S}'|\le m + l\log m = \tilde{\mathcal{O}}(m + l).
		\end{equation}
		\item $\mathcal{S}\subset\mathcal{S}'$ and the value function $V$ at the equilibrium of the constructed 2-TBSG satisfies:
		\begin{equation}\label{eqin2}
		V(s) = c\cdot v(s), \quad \forall s\in\mathcal{S},
		\end{equation}
		where $v$ is the equilibrium value function of the original 2-TBSG, and
		\begin{equation}
			c = \gamma^{\frac{\lceil\log m\rceil - 1}{\lceil\log m\rceil}}.
		\end{equation}
	\end{itemize}
\end{theorem}

\begin{proof}
\par Our proof consists of two parts. In the first part, we construct a new 2-TBSG where each state has no more than two actions, and the value function at equilibrium of original 2-TBSG can be easily obtained given the equilibrium value of the constructed 2-TBSG (proportional to the value at some states in the constructed 2-TBSG). In the second part, we modify the constructed 2-TBSG so that each state has exactly two actions, while keeping the equilibrium value unchanged by constructing an obvious undesirable action for those states with only one action. 
\par We first construct a binary tree rooted at $s$ with exactly $|A_{s}|$ leaves, and the depth of the tree is exactly $p = \lceil\log m\rceil$. This tree is called the depth-$p$ binary tree of state $s$:
\begin{itemize}
	\item In the first $p - \log\lceil|A_{s}|\rceil$ layers, each node has only one child.
	\item In the last $\log\lceil|A_{s}|\rceil$ layers, it is a binary tree with exactly $|A_{s}|$ leaves.
	\item Every leaves has depth $p$.
\end{itemize}
\par Each node except the root $s$ and all leaves in the depth-$p$ binary tree of $s$ are assigned with a new state whose owner is same as state $s$ (player 1 or player 2). We use $S_{1}$ to denote the set of states in the first $p - 2$ layers, and $S_{2}$ to denote the set of states in the $(p-1)$-th layer. The parameters (transition probabilities, rewards, discounted factor) are given as follows:
\begin{itemize}
	\item For each state in $S_{1}$, one or two actions are assigned to it depending on how many children states (its children in the binary tree) it has, with probability 1 leading to a child state and reward 0.
	\item For set $S_{2}$, each of their children nodes is assigned with an action of $s$ in the original 2-TBSG. This can be done since the total number of children nodes of $S_{2}$ is exactly $|\mathcal{A}_{s}|$. For each state in $S_{2}$, its actions are given by its children nodes. The transition probability and reward of taking that action is assigned to be the same as in the original 2-TBSG.
	\item The discounted factor in the constructed 2-TBSG is given by $\delta = \gamma^{1/p}$.
\end{itemize}
A special case of can be viewed in Figure \ref{fig1} when $p = 4, |\mathcal{A}_{s}| = 7$.
\par It is easy to obtain that the number of states in the constructed 2-TBSG is no more than 
\begin{equation*}
	m + n\log m.
\end{equation*}
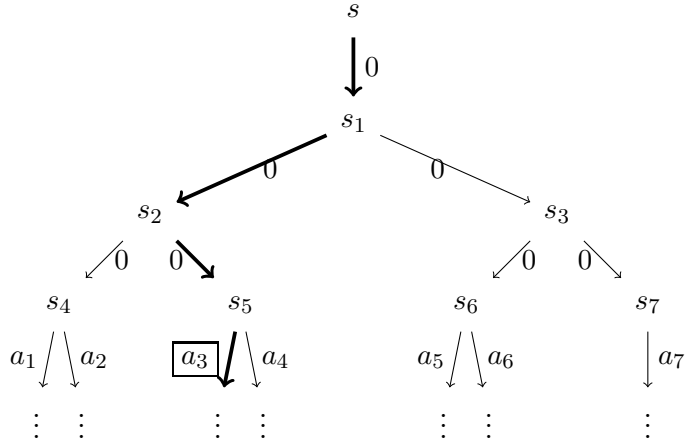
\begin{figure}
	\centering
	\setlength{\belowcaptionskip}{10pt}
	\caption{Example when $t=4$}
	\label{fig1}
	\tikzstyle{rec} = [rectangle, minimum width = 0.7cm, minimum height = 0.7cm, text centered]
    \begin{tikzpicture}
    	\node(rec0)[rec]{$s$};
    	\node(rec1)[rec, below of = rec0, yshift = -0.5cm]{$s_{1}$};
		\node(rec2)[rec, below left of = rec1, xshift = -2cm, yshift = -0.5cm]{$s_{2}$};
		\node(rec3)[rec, below right of = rec1, xshift = 2cm, yshift = -0.5cm]{$s_{3}$};
		\node(rec4)[rec, below left of = rec2, xshift = -0.5cm, yshift = -0.5cm]{$s_{4}$};
		\node(rec5)[rec, below right of = rec2, xshift = 0.5cm, yshift = -0.5cm]{$s_{5}$};
		\node(rec6)[rec, below left of = rec3, xshift = -0.5cm, yshift = -0.5cm]{$s_{6}$};
		\node(rec7)[rec, below right of = rec3, xshift = 0.5cm, yshift = -0.5cm]{$s_{7}$};
		\node(rec8)[rec, below of = rec4, xshift = -0.3cm, yshift = -0.5cm]{$\vdots$};
		\node(rec9)[rec, below of = rec4, xshift = 0.3cm, yshift = -0.5cm]{$\vdots$};
		\node(rec10)[rec, below of = rec5, xshift = -0.3cm, yshift = -0.5cm]{$\vdots$};
		\node(rec11)[rec, below of = rec5, xshift = 0.3cm, yshift = -0.5cm]{$\vdots$};
		\node(rec12)[rec, below of = rec6, xshift = -0.3cm, yshift = -0.5cm]{$\vdots$};
		\node(rec13)[rec, below of = rec6, xshift = 0.3cm, yshift = -0.5cm]{$\vdots$};
		\node(rec14)[rec, below of = rec7, yshift = -0.5cm]{$\vdots$};
		\draw[->, line width=0.5mm] (rec0)--(rec1) node[midway, right] {0};
		\draw[->, line width=0.5mm] (rec1)--(rec2) node[midway, right] {0};
		\draw[->] (rec1)--(rec3) node[midway, left] {0};;
		\draw[->] (rec2)--(rec4) node[midway, right] {0};;
		\draw[->, line width=0.5mm] (rec2)--(rec5) node[midway, left] {0};;
		\draw[->] (rec3)--(rec6) node[midway, right] {0};;
		\draw[->] (rec3)--(rec7) node[midway, left] {0};;
		\draw[->] (rec4)--(rec8) node[midway, left] {$a_{1}$};
		\draw[->] (rec4)--(rec9) node[midway, right] {$a_{2}$};
		\draw[->, line width=0.5mm] (rec5)--(rec10) node[midway, left] {$\boxed{a_{3}}$};
		\draw[->] (rec5)--(rec11) node[midway, right] {$a_{4}$};
		\draw[->] (rec6)--(rec12) node[midway, left] {$a_{5}$};
		\draw[->] (rec6)--(rec13) node[midway, right] {$a_{6}$};
		\draw[->] (rec7)--(rec14) node[midway, right] {$a_{7}$};
    \end{tikzpicture}
\end{figure}
\par We next present a definition of final actions and the executing path of a state.
\begin{definition}
\par For a given strategy $\pi'$ in the constructed 2-TBSG cases and $s\in\mathcal{S}$, we continue the following process:
\begin{itemize}
	\item $s_{0}\leftarrow s,\quad i\leftarrow 0$;
	\item If $\pi'(s_{n})$ is a constructed action (not an action in the
    original 2-TBSG), then we let $s_{n+1}$ to be the state obtained by
executing action $\pi'(s_{n})$. Since all constructed actions are deterministic,
there is only one choice of $s_{n+1}$. Then let $n\leftarrow n+1$.
	\item If $\pi'(s_{n})$ is an action in the original 2-TBSG, then we stop this process, and call $\pi'(s_{n})\in\mathcal{A}$ to be the final action of $s$, and path $s_{0}\to s_{1}\to\cdots\to s_{n}\to \pi'(s_{n})$ to be the executed path from $s$ to action $\pi'(s_{n})$.
\end{itemize}
\end{definition}
\par Notice that the previous described process must be ended in $p-1$ steps, and all states in the executed path of $s$ must lie in the depth-$p$ binary tree of $s$. For any state $s\in\mathcal{S}$ and $a\in\mathcal{A}_{s}$, there exists a unique executed path from $s$ to $a$. In Figure \ref{fig1} we present an example of final actions and executed paths. When the strategy $\pi'$ follows bold arrows, the final action of $s$ will be $a_{3}$, and the executed path from $s$ to $a_{3}$ is $s\to s_{1}\to s_{2}\to s_{5}\to a_{3}$.
\par Based on the final actions, we define the corresponding strategy $\pi$ with respect to $\pi'$ in the original 2-TBSG: for each state $s\in\mathcal{S}$, $\pi(s)$ is defined to be the final action of $s$ in $\pi'$. Next, we prove that for any state $s\in\mathcal{S}$, the value of $s$ in strategy $\pi'$ agrees with $\delta^{p-1}$ times the value of $s$ in strategy $\pi$. Actually, along the trajectory of $\pi'$, we meet a final action every $p$ steps, and only final actions have nonzero rewards. Hence values of $s$ in $\pi$ and $\pi'$ satisfy
\begin{equation*}
	V^{\pi'}(s) = \mathbb{E}_{\pi'}\sum_{i=1}^{\infty}r_{a}\cdot\delta^{pi-1} = \delta^{p-1}\mathbb{E}_{\pi}\sum_{i=1}^{\infty}r_{a}\cdot\gamma^{i-1} = \delta^{p-1}v^{\pi}(s),
\end{equation*}
where $a'$ denotes actions along strategy $\pi'$, and $a$ denotes actions along strategy $\pi$.
\par What is left in the proof is to show that if $\pi' = (\pi_{1}', \pi_{2}')$ is an equilibrium strategy in the constructed 2-TBSG, then $\pi = (\pi_{1}, \pi_{2})$ is an equilibrium strategy in the original 2-TBSG. 
 For any player 1's state $s$ and action $a\in\mathcal{A}_{s}$, we use $\eta =
 (\eta_{1}, \pi_{2})$ to denote the strategy in the original 2-TBSG: 
\begin{equation*}
	\eta_{1}(s_{1}) = \begin{cases} a &\quad \text{if }s_{1} = s,\\\pi_{1}(s_{1}) &\quad \text{otherwise}.\end{cases}
\end{equation*}
In the constructed 2-TBSG, there exists a unique executed path $T$ from $s$ to action $a$, and for any state $s_{1}$ on this path $T$, there is only one action $\tau(s_{1})$ in $\mathcal{A}'_{s_{1}}$ such that the next state when using $\tau(s_{1})$ also lies on $T$. We define player 1's strategy $\eta' = (\eta'_{1}, \pi_{2})$ as follows:
\begin{equation*}
	\eta'_{1}(s_{1}) = \begin{cases}
		\tau(s_{1})&\quad \text{if }s_{1}\in T,\\
		\lambda(s_{1})&\quad \text{if } s_{1} \text{ is in the depth-$p$ binary tree of } s \text{ and } s_{1}\not\in T,\\
		\pi_{1}'(s_{1})&\quad \text{if } s_{1} \text{ is not in the depth-$p$ binary tree of } s,
	\end{cases}
\end{equation*}
where $\lambda(s_{1})$ can be chosen $\mathcal{A}'_{s_{1}}$ arbitrarily. Then it is easy to examine that $\eta$ is the corresponding strategy of $\eta'$. Since $\eta$ is an equilibrium strategy of the constructed 2-TBSG, we have
\begin{equation*}
	v^{\eta} = \delta^{-p+1}V^{\eta'}\le \delta^{-p+1}V^{\pi'} = v^{\pi},
\end{equation*}
where the inequality is due to the property of equilibrium strategy. Furthermore, according to Lemma \ref{several} and Lemma \ref{condition}, we have $x_{a}^{\eta}r^{\pi}_{a} = \mathbf{1}^{T}(v^{\eta} - v^{\pi})\ge 0$, where $x_{a}^{\eta}\ge 1$. This indicates that $r^{\pi}_{a}\le 0$. Since $a$ can be chosen arbitrarily, we have $(r^{\pi})_{\mathcal{A}_{1}}\le 0$, and similarly $(r^{\pi})_{\mathcal{A}_{2}}\ge 0$. Again according to Lemma \ref{condition}, we obtain that $\pi$ is an equilibrium strategy of the original 2-TBSG.

\par Next, we handle states with only one actions. For each of such state, our technique is to construct another action which is obviously unacceptable to appear in equilibrium strategies. 
 For state $s$ with only one action $a_{1}$, we assign it with another action $a_{2}$:
\begin{itemize}
	\item The transition probability using action $a_{2}$ is identical to $a_{1}$.
	\item If state $s$ belongs to player 1, then the reward of $a_{2}$ is
        assigned to be smaller than the reward of $a_{1}$.
	\item If state $s$ belongs to player 2, then the reward of $a_{2}$ is
        assigned to be larger than the reward of $a_{1}$.
\end{itemize}
\par If we construct actions in such ways, it is obvious that action $a_{2}$ is
inferior to $a_{1}$ according to its owner (player 1 or player 2). Hence any
strategy which possesses $a_{2}$ is not an equilibrium strategy, since switching
action $a_{2}$ into $a_{1}$ leads to a better strategy for its owner.
Combining these two parts together proves Theorem \ref{2-TBSG}.
\end{proof}
\begin{remark}
	Since it is easy to obtain the equilibrium strategy from the equilibrium value and vice versa, we can solve the original 2-TBSG by solving the constructed 2-TBSG. 
\end{remark}

\section{Conclusion}
 In this paper, we propose two different algorithms for 2-TBSG with strongly
polynomial complexity: the simplex strategy iteration algorithm and the modified
simplex strategy iteration algorithm.  We propose a class of geometrically
converging algorithms and develop a proof technique to prove the strongly polynomial
complexity when the discounted factor is fixed. Furthermore, we present how to
transform a general 2-TBSG into a special 2-TBSG where each state has exactly
two actions. Specifically, our simplex strategy iteration algorithm is coincident with
the simplex method in the MDP cases.
 These analysis and properties shed some light on the open problem of solving the
deterministic 2-TBSG in strongly polynomial time independent of the discount
factor. 

\bibliography{reference}

\end{document}